\documentclass[11pt]{article}

\usepackage{subfigure}
\usepackage{xspace}

\usepackage[left=1in, right=1in, top=1in, bottom=1in]{geometry}
\usepackage[group-separator={,}]{siunitx}
\usepackage{amsmath,amsfonts,amssymb, amsthm}
\usepackage{graphicx}
\usepackage{url}
\usepackage{mathtools}
\usepackage[textsize=tiny,textwidth=1cm,shadow]{todonotes}
\usepackage{thmtools}
\usepackage{enumitem}
\usepackage{pdflscape}
\usepackage{verbatim}
\usepackage{nccmath}
\usepackage[utf8]{inputenc}
\usepackage{array}
\definecolor{MyBlue}{rgb}{0.12, 0.12, 0.76}

\usepackage{algorithm}
\usepackage{algpseudocode}
\usepackage{algorithmicx}
\let\oldReturn\Return
\renewcommand{\Return}{\State\oldReturn}
\algtext*{EndWhile}
\algtext*{EndIf}
\algtext*{EndForAll}
\algtext*{EndFor}
\algtext*{EndFunction}

\usepackage{pifont}
\usepackage{pgf}
\usepackage{tikz}
\usetikzlibrary{shadows,arrows,decorations,decorations.shapes,backgrounds,shapes,snakes,automata,fit,petri,shapes.multipart,calc,positioning,shapes.geometric,graphs,graphs.standard}

\makeatletter
\newcommand{\thickhline}{%
    \noalign {\ifnum 0=`}\fi \hrule height 1.4pt
    \futurelet \reserved@a \@xhline
}
\newcolumntype{"}{@{\hskip\tabcolsep\vrule width 1.4pt\hskip\tabcolsep}}
\makeatother

\allowdisplaybreaks[1] 

\newtheorem{theorem}{Theorem}[section]
\newtheorem{lemma}{Lemma}[section]
\newtheorem{corollary}{Corollary}[theorem]
\newtheorem{definition}{Definition}[section]

\DeclareMathOperator*{\argmax}{arg\,max}

\usepackage{natbib}
\usepackage{tablefootnote} 
\usepackage{epsfig}

\newtheorem{property}	[theorem]	{Property}

\newlist{exlist}{enumerate}{1}
\setlist[exlist]{label=(\alph*)}

\newcommand{\eps}{\epsilon}

\newcommand{\sm}{\setminus}

\newcommand{\td}{{$\delta$-triangle-dense}\xspace}
\newcommand{\trans}{\tau}

\begin{document}

\title{Distribution-Free Models of Social Networks\thanks{Chapter~28 of the book {\em Beyond the
      Worst-Case Analysis of Algorithms}~\citep{bwca}.}}
\author{Tim Roughgarden\thanks{Department of Computer Science,
    Columbia University.  
Supported in part by NSF award
    CCF-1813188 and ARO award W911NF1910294.
Email: \texttt{tim.roughgarden@gmail.com.}} \and
C. Seshadhri\thanks{Department of Computer Science, University of
  California at Santa Cruz.  Supported in part by NSF TRIPODS
grant CCF-1740850, NSF grants CCF-1813165 and CCF-1909790, and ARO
award W911NF1910294.  Email: \texttt{sesh@ucsc.edu}.}}

\maketitle

\begin{abstract}
  The structure of large-scale social networks has
  predominantly been articulated using generative models, a form of
  average-case analysis.  This chapter surveys recent proposals of
  more robust models of such networks. These models
  posit deterministic and
  empirically supported combinatorial structure rather than a specific
  probability
  distribution. We discuss the formal definitions
  of these models and how they relate to empirical observations in
  social networks, as well as the known structural and algorithmic
  results for the corresponding graph classes.
\end{abstract}

\section{Introduction}\label{s:intro}

Technological developments in the 21st century have given rise to
large-scale social networks, such as the graphs defined by Facebook
friendship relationships or followers on Twitter.  Such 
networks arguably provide the most important new
application domain for graph analysis in well over a decade.

\subsection{Social Networks Have Special Structure}

There is wide consensus that social networks have predictable
structure and features, and accordingly are not well modeled by
arbitrary graphs.  
From a structural viewpoint, the most well studied and empirically
validated properties of social networks are:
\begin{enumerate}

\item A heavy-tailed degree 
distribution, such as a power-law distribution.

\item Triadic closure, meaning that pairs of vertices with a common
  neighbor tend to be directly connected---that friends of friends
  tend to be friends in their own right.

\item The presence of ``community-like structures,'' meaning subgraphs
  that are much more richly connected internally than externally.

  \item The small-world property, meaning that it's possible to travel
    from any vertex to any other vertex using remarkably few hops.

\end{enumerate}
These properties are not generally possessed by Erd\H{o}s-R\'{e}nyi random
graphs (in which each edge is present independently with some
probability~$p$); a new model is needed to capture them.

From an algorithmic standpoint, empirical results indicate that
optimization problems are often easier to solve in social networks
than in worst-case graphs.  
For example, lightweight
heuristics are unreasonably effective in practice for finding the
maximum clique or recovering dense subgraphs of a large social network.

The literature on models that capture the special structure of social
networks is almost entirely driven by the quest for generative (i.e.,  probabilistic) models
that replicate some or all of the four properties listed above.
Dozens of generative models have been proposed, and
there is little consensus about which is the
``right'' one.
The plethora of models poses a challenge to 
meaningful theoretical work on social networks---which of the models,
if any, is to be believed?  How can we be sure that a given
algorithmic or structural result is not an artifact of
the model chosen?

This chapter surveys recent research on more robust models of
large-scale social 
networks, which assume
deterministic combinatorial properties rather than a specific
generative model.  Structural and algorithmic results that rely only
on these deterministic properties automatically carry over to any
generative model that produces graphs possessing these properties (with
high probability).   Such results effectively apply ``in 
the worst case over all plausible generative models.''  
This hybrid of worst-case (over input distributions) and average-case
(with respect to the distribution) analysis resembles
several of 
the semi-random models discussed elsewhere in the book,
such as in the
preceding chapters on pseudorandom data (Chapter~26) and
prior-independent auctions (Chapter~27).

Sections~\ref{s:cclosed} and~\ref{s:td} of this chapter cover two
models of social networks that are motivated by triadic closure, the
second of the four signatures of social networks listed in
Section~\ref{s:intro}.  Sections~\ref{s:plb} and~\ref{s:bct} discuss
two models motivated by heavy-tailed degree distributions.

\section{Cliques of $c$-Closed Graphs}\label{s:cclosed}

\subsection{Triadic Closure}

Triadic closure is the property that, when two members of a social
network have a friend in common, they are likely to be friends
themselves.  In graph-theoretic terminology, two-hop paths tend to
induce triangles.

Triadic closure has been studied for decades in the social
sciences 
and there is compelling intuition for why social networks should exhibit
strong triadic closure properties.  Two people with a common friend
are much more likely to meet than two arbitrary people, and are likely
to share common interests.  They might also feel pressure to be friends
to avoid imposing stress on their relationships with their common
friend.

The data support this intuition.  Numerous large-scale studies on
online social networks
provide overwhelming empirical evidence for triadic closure.  The plot
in Figure~\ref{f:enron}, derived from the network of email
communications at the disgraced energy company Enron, is
representative.  Other social networks exhibit similar triadic closure
properties.

\begin{figure}[t]
\begin{center}
\mbox{\subfigure[Triadic closure in the Enron email network]{\epsfig{file=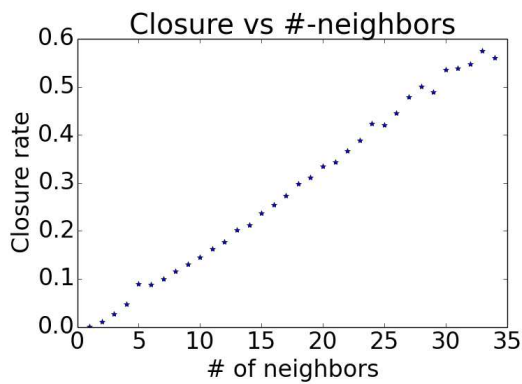,width=.45\textwidth}}\qquad
\subfigure[Triadic closure in a random
graph]{\epsfig{file=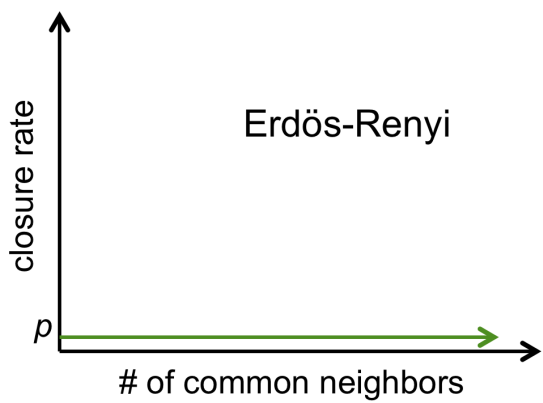,width = .45\textwidth}}}
\caption{In the Enron email graph,
vertices correspond to Enron employees, and there is an edge
  connecting two employees if one sent at least
  one email to the other.  In~(a),
vertex pairs of this graph are grouped according to the
number of common neighbors (indicated on the $x$-axis).  The
  $y$-axis shows the fraction of such pairs that are themselves
  connected by an edge.  The edge density---the fraction of arbitrary
  vertex pairs that are directly connected---is roughly~$10^{-4}$.
In~(b), a cartoon of the analogous plot for an
Erd\H{o}s-R\'{e}nyi graph with edge density~$p=10^{-4}$ is
shown. Erd\H{o}s-R\'{e}nyi graphs are not a good model for networks
like the Enron network---their closure rate is too small, and the
closure rate fails to increase
as the number of common neighbors increases.}
\label{f:enron}
\end{center}
\end{figure}

\subsection{$c$-Closed Graphs}

The most extreme version of triadic closure would assert that whenever
two vertices have a common neighbor, they are themselves neighbors:
whenever $(u,v)$ and $(v,w)$ are in the edge set~$E$, so is
$(u,w)$.  The class of graphs satisfying this property is not very
interesting---it is precisely the (vertex-)disjoint unions of
cliques---but it forms a natural base case for more interesting
parameterized definitions.\footnote{Recall that a {\em clique} of a
  graph~$G=(V,E)$ is a subset $S \subseteq V$ of vertices that are
  fully connected, meaning that
  $(u,v) \in E$ for every pair $u,v$ of distinct vertices of~$S$.}

Our first definition of a class of graphs with strong triadic closure
properties is that of {\em $c$-closed graphs}.
\begin{definition}[\citet{cclosed}] \label{d:closed}
For a positive integer~$c$, a graph $G=(V,E)$ is {\em
  $c$-closed} if, whenever $u,v \in V$ have at least $c$ common
neighbors, $(u,v) \in E$.  
\end{definition}
For a fixed number of vertices, the parameter~$c$ interpolates between
unions of cliques (when $c=1$) and all graphs (when $c=|V|-1$).  The
class of 2-closed graphs---the graphs that do not contain a square
(i.e., $K_{2,2}$) or a diamond (i.e., $K_4$ minus an edge) as an
induced subgraph---is already non-trivial.  The $c$-closed condition
is a coarse proxy for the empirical closure rates observed in social
networks (like in Figure~\ref{f:enron}), asserting that the closure
rate jumps to 100\% for vertices with~$c$ or more common neighbors.

Next is a less stringent version of the definition, which is
sufficient for the main algorithmic result of this section.
\begin{definition}[\citet{cclosed}] \label{d:weak}
For a positive integer~$c$, a vertex $v$ of a graph $G=(V,E)$ is 
{\em $c$-good} if whenever $v$ has at least $c$ common
neighbors with another vertex~$u$, $(u,v) \in E$.  
The graph~$G$ is {\em weakly $c$-closed} if every induced subgraph has
at least one $c$-good vertex.
\end{definition}
A $c$-closed graph is also weakly $c$-closed, as each of its vertices
is $c$-good in each of its induced subgraphs.  The converse is false;
for example, a path graph is not 1-closed, but it is weakly 1-closed
(as the endpoints of a path are 1-good).
Equivalent to
Definition~\ref{d:weak} is the condition that the graph~$G$ has an
elimination ordering of $c$-good vertices, meaning the vertices can be
ordered $v_1,v_2,\ldots,v_n$ such that, for every~$i=1,2,\ldots,n$,
the vertex $v_i$ is $c$-good in the subgraph induced by
$v_i,v_{i+1},\ldots,v_n$ (Exercise~\ref{exer:elim}).
Are real-world social networks $c$-closed or weakly $c$-closed for
reasonable values of~$c$?  The next table summarizes 
some representative numbers.

\begin {table}[ht]
\begin{center}
\begin{tabular}{| l | l | l | l | l | }
  		\hline
  &$n$&	$m$&	$c$&	weak $c$ \\
  \hline
  email-Enron&	36692&	183831&	161&	34 \\
  \hline
  p2p-Gnutella04	&10876	&39994	&24&	8 \\
  \hline  
wiki-Vote&7115	&103689&	420&	42\\
\hline 
ca-GrQc&5242	&14496&	41&	9\\
\hline

\end{tabular}
\end{center}
\caption{The $c$-closure and weak $c$-closure of four well-studied
  social networks from the SNAP (Stanford Large Network Dataset) collection of
  benchmarks (\protect\url{http://snap.stanford.edu/}).
``email-Enron'' is the network described in Figure~\ref{f:enron};
``p2p-Gnutella04'' is the topology of a Gnutella peer-to-peer network
circa 2002; ``wiki-Vote'' is the network of who votes on whom in 
promotion cases on Wikipedia; and ``ca-GrQc'' is the collaboration
network of authors of papers uploaded to the General Relativity and
Quantum Cosmology section of arXiv.
  For each network $G$, $n$ indicates the number of vertices, $m$ the number of edges, $c$ the smallest value $\gamma$ such that $G$ is $\gamma$-closed, and ``weak $c$'' the smallest value $\gamma$ such that $G$ is weakly $\gamma$-closed.}\label{table:numbers}
\end{table}

These social networks are $c$-closed for much
smaller values of~$c$ than the trivial bound of~$n-1$, and are weakly
$c$-closed for quite modest values of~$c$.

\subsection{Computing a Maximum Clique: A Backtracking Algorithm}\label{ss:backtracking}

Once a class of graph has been defined, such as $c$-closed graphs, a
natural agenda is to investigate
fundamental optimization problems with graphs restricted to
the class.  We single out the problem of finding the maximum-size
clique of a graph, primarily because it is one of the most central
problems in social network analysis.  In a social network, cliques can
be interpreted as the most extreme form of a community.

The problem of computing the maximum clique of a graph reduces to the
problem of enumerating the graph's maximal cliques\footnote{A maximal clique is a clique that is not a strict subset of another clique.}---the maximum
clique is also maximal, so it appears as the largest of the cliques
in the enumeration.  

How does the $c$-closed condition help with the efficient computation of a
maximum clique?
We next observe that the problem of reporting all maximal cliques is
polynomial-time solvable in $c$-closed graphs when~$c$ is a fixed
constant.  The algorithm is based on backtracking. For convenience,
we give a procedure that, for any
vertex~$v$, identifies all maximal cliques that contain $v$. (The full procedure
loops over all vertices.)
\begin{enumerate}

\item Maintain a history $H$, initially empty.

\item Let $N$ denote the vertex set comprising $v$ and all
vertices $w$ that are adjacent to both $v$ and all vertices in $H$.

\item If $N$ is a clique, report the clique $H \cup N$ and
  return.

\item Otherwise, recurse on each vertex~$w \in N \setminus \{v\}$ with
history $H := H \cup \{v\}$.  

\end{enumerate}
This subroutine reports all maximal cliques that contain $v$, whether
the graph is $c$-closed or not (Exercise~\ref{exer:maximal}).  In a
$c$-closed graph, the maximum depth of the recursion is~$c$---once
$|H|=c-1$, every pair of vertices in $N \setminus \{v\}$ has $c$
common neighbors (namely $H \cup \{v\}$) and hence $N$ must be a
clique.  The running time of the backtracking algorithm is therefore
$n^{c+O(1)}$ in $c$-closed graphs.

This simplistic backtracking algorithm is extremely slow except for
very small values of~$c$.  Can we do better?

\subsection{Computing a Maximum Clique: Fixed-Parameter Tractability}

There is a simple but clever algorithm that, for an arbitrary graph,
enumerates all of the maximal cliques while using only polynomial time
per clique.
\begin{theorem}[\citet{T+77}]\label{t:enum}
There is an algorithm that, given any input graph with $n$ vertices
and $m$ edges, outputs all of the maximal cliques of the graph
in~$O(mn)$ time per maximal clique.
\end{theorem}
Theorem~\ref{t:enum} reduces the problem of enumerating all maximal
cliques in polynomial time to the combinatorial task of proving
a polynomial upper bound on the number of maximal cliques.  

Computing a maximum clique of an arbitrary graph is an $NP$-hard
problem, so presumably there exist graphs with an exponential number
of maximal cliques.  The {\em Moon-Moser graphs} are a simple and
famous example.  For~$n$ a multiple of~3, the Moon-Moser graph
with~$n$ vertices is the perfectly balanced $\tfrac{n}{3}$-tite graph,
meaning the vertices are partitioned into $\tfrac{n}{3}$ groups of~3,
and every vertex is connected to every other vertex except for the~2
vertices in the same group (Figure~\ref{f:mm}).
\begin{figure}[t]
\begin{center}
\includegraphics[width=.5\textwidth]{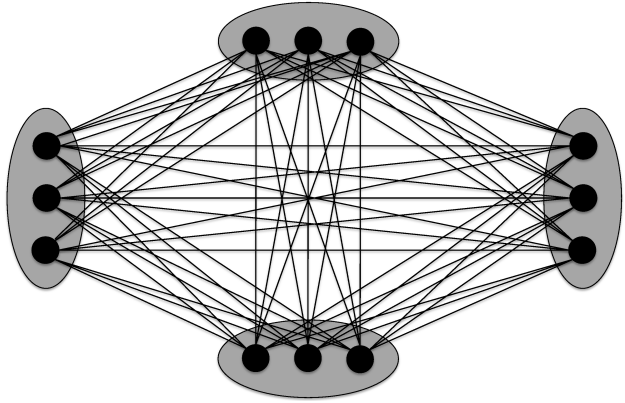}
\caption{The Moon-Moser graph with~$n=12$ vertices.}\label{f:mm}
\end{center}
\end{figure}
Choosing one vertex from each group induces a maximal clique, for a
total of $3^{n/3}$ maximal cliques, and these are all of the maximal
cliques of the graph.  More generally, a basic result in graph theory
asserts that {\em no} $n$-vertex graph can
have more than~$3^{n/3}$ maximal cliques.
\begin{theorem}[\citet{MM65}]\label{t:mm}
Every $n$-vertex graph has at most $3^{n/3}$ maximal cliques.
\end{theorem}
A Moon-Moser graph on $n$ vertices is not $c$-closed even for $c=n-3$,
so there remains hope for a positive result for $c$-closed graphs with
small~$c$.  The Moon-Moser graphs do show that the number of maximal
cliques of a $c$-closed graph can be exponential in $c$ (since a
Moon-Moser graph on $c$ vertices is trivially $c$-closed).  Thus the
best-case scenario for enumerating the maximal cliques of a $c$-closed
graph is a fixed-parameter tractability result (with respect to the
parameter~$c$), stating that, for some function~$f$ and constant~$d$
(independent of $c$), the number of maximal cliques in an $n$-vertex
$c$-closed graph is $O(f(c) \cdot n^d)$. The next theorem shows that this is
indeed the case, even for weakly $c$-closed graphs.

\begin{theorem}[\citet{cclosed}] \label{t:c-closed} 
Every weakly $c$-closed graph with $n$ vertices has at most 
\[
3^{(c-1)/3} \cdot n^2
\] maximal cliques.
\end{theorem}

The following corollary is immediate from Theorems~\ref{t:enum} and~\ref{t:c-closed}.
\begin{corollary}
The maximum clique problem is polynomial-time solvable in
weakly $c$-closed $n$-vertex graphs with~$c = O(\log n)$.
\end{corollary}

\subsection{Proof of Theorem~\ref{t:c-closed}}

The proof of Theorem~\ref{t:c-closed} proceeds by induction on the
number of vertices~$n$.  (One of the factors of~$n$ in the bound is
from the~$n$
steps in this induction.)  Let~$G$ be an $n$-vertex weakly $c$-closed
graph.  Assume that $n \ge 3$; otherwise, the bound is trivial.

By assumption, $G$ has a $c$-good vertex~$v$.  By induction,
$G \sm \{v\}$ has at most $(n-1)^2 \cdot 3^{(c-1)/3}$ maximal cliques.
(An induced subgraph of a weakly $c$-closed graph is again weakly
$c$-closed.)  Every maximal clique~$C$ of $G \sm \{v\}$ gives rise to
a unique maximal clique in~$G$ (namely $C$ or $C \cup \{v\}$,
depending on whether the latter is a clique).  It remains to bound the
number of {\em uncounted} maximal cliques of~$G$, meaning the maximal
cliques~$K$ of~$G$ for which $K \sm \{v\}$ is not maximal
in~$G \sm \{v\}$.

An uncounted maximal clique $K$ must include~$v$, with $K$ contained in $v$'s
neighborhood (i.e., in the subgraph induced by $v$ and the vertices
adjacent to it).  Also, there must be a vertex $u \notin K$ such that
$K \sm \{v\} \cup \{u\}$ is a clique in~$G \sm \{v\}$; we say that~$u$
is a {\em witness} for~$K$, as it certifies the non-maximality
of~$K \sm \{v\}$ in $G \sm \{v\}$.  Such a witness must be connected
to every vertex of $K \sm \{v\}$.  It cannot be a neighbor of~$v$, as
otherwise $K \cup \{u\}$ would be a clique in~$G$, contradicting $K$'s
maximality.

Choose an arbitrary witness for each uncounted clique
of~$G$ and bucket these cliques according to their witness; recall
that all witnesses are non-neighbors of~$v$.
For every uncounted clique~$K$ with witness~$u$, all vertices of the clique
$K \sm \{v\}$ are connected to both $v$ and $u$.  Moreover,
because~$K$ is a maximal clique in~$G$, $K \sm \{v\}$ is a maximal
clique in the subgraph~$G_u$ induced by the common neighbors of $u$
and~$v$.

\begin{figure}[t]
\begin{center}
\includegraphics[width=.5\textwidth]{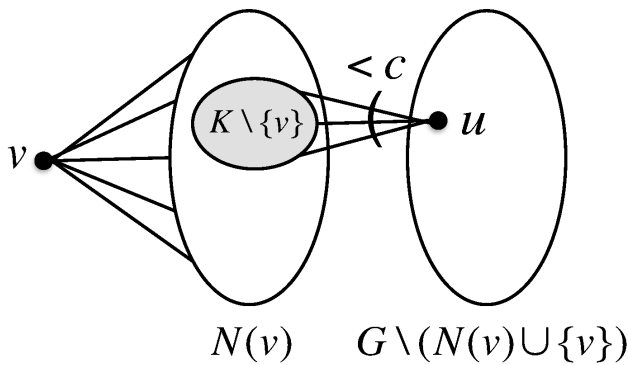}
\caption{Proof of Theorem~\ref{t:c-closed}.  $N(v)$ denotes the
  neighbors of~$v$.  $K$ denotes a maximal clique of~$G$ such that $K
  \sm \{v\}$ is not maximal in $G \sm \{v\}$.  There is a vertex~$u$,
  not connected to~$v$, that witnesses the non-maximality of $K \sm
  \{v\}$ in $G \sm \{v\}$.  Because $v$ is a $c$-good vertex,
  $u$ and $v$ have at most $c-1$ common neighbors.}\label{f:cclosed}
\end{center}
\end{figure}

How big can such a subgraph~$G_u$ be?  This is the step of the proof
where the weakly $c$-closed condition is important: Because $u$ is a
non-neighbor of~$v$ and $v$ is a $c$-good vertex, $u$ and $v$ have at
most $c-1$ common neighbors and hence~$G_u$ has at most $c-1$ vertices
(Figure~\ref{f:cclosed}).  
By the Moon-Moser theorem
(Theorem~\ref{t:mm}), each subgraph~$G_u$ has at most $3^{(c-1)/3}$
maximal cliques.  Adding up over the at most $n$ choices for~$u$, the
number of uncounted cliques is at most $n \cdot 3^{(c-1)/3}$; this sum
over possible witnesses is the source of the second factor of~$n$ in
Theorem~\ref{t:c-closed}.  Combining this bound on the uncounted cliques
with the inductive bound on the remaining maximal cliques of~$G$
yields the desired upper bound of
\[
(n-1)^2 \cdot 3^{(c-1)/3} + n \cdot 3^{(c-1)/3} \le n^2 \cdot 3^{(c-1)/3}.
\]

\section{The Structure of Triangle-Dense Graphs}\label{s:td}

\subsection{Triangle-Dense Graphs}

Our second graph class inspired by the strong triadic closure
properties of social and information networks is the class of {\em
  \td} graphs. These are graphs where a constant fraction of vertex
pairs having at least one common neighbor are directly connected by
an edge.  Equivalently, a constant fraction of the wedges (i.e.,
two-hop paths) of the graph belong to a triangle.

\begin{definition}[\citet{GRSjournal}] \label{d:dense}
The {\em triangle density} of an undirected graph~$G$ is
$\trans(G) := 3t(G)/w(G)$, where $t(G)$ and $w(G)$ denote the number
of triangles and wedges of~$G$, respectively.
(We define $\trans(G) = 0$ if $w(G)=0$.)
The class of {\em \td graphs} consists of the graphs~$G$
with $\trans(G) \geq \delta$.
\end{definition}

(In the social networks literature, this is also called the \emph{transitivity}
or the \emph{global clustering coefficient}.)
Because every triangle of a graph contains~3 wedges, and no two
triangles share a wedge, the triangle density of a graph is between~0
and~1---the fraction of wedges that belong to a
triangle.
Triangle density is another coarse proxy for the empirical
closure rates observed in social networks (like in
Figure~\ref{f:enron}(a)).

The 1-triangle-dense graphs are precisely the unions of disjoint
cliques, while triangle-free graphs constitute the 0-triangle-dense graphs.
The triangle density of an Erd\H{o}s-R\'{e}nyi graph with edge
probability~$p$ is concentrated around~$p$ (cf.,
Figure~\ref{f:enron}(b)).  For an 
Erd\H{o}s-R\'{e}nyi graph to have constant
triangle density, one would need to set $p = \Omega(1)$. This would imply that the graph is dense,
quite unlike social networks.
For example, in the
year~2011 the triangle density of the Facebook graph was computed to
be~$0.16$, which is five orders of magnitude larger than in a random
graph with the same number of vertices (roughly 1 billion at the time)
and edges (roughly 100 billion).

\subsection{Visualizing Triangle-Dense Graphs}

What do \td graphs look like?  Can we make any structural assertions
about them, akin to separator theorems for planar graphs (allowing
them to be viewed as ``approximate grids'') or the
regularity lemma for dense graphs (allowing them to viewed as
approximate unions of random bipartite graphs)?

Given that 1-triangle-dense graphs are unions of cliques, a first
guess might be that \td graphs look like the approximate union of
approximate cliques (as in Figure~\ref{f:inverse}(a)).  Such graphs
certainly have high triangle density; could there be an ``inverse
theorem,'' stating that these are in some sense the {\em only} graphs
with this property?

In its simplest form, the answer to this question is ``no,''
as \td graphs become quite diverse once $\delta$ is bounded below~1.
For example, adding a clique on $n^{2/5}$ vertices
to an arbitrary bounded-degree $n$-vertex graph produces a \td graph
with $\delta = 1 - o(1)$ as $n \rightarrow \infty$ (see  Figure~\ref{f:inverse}(b)).

\begin{figure}[t]
\begin{center}
\mbox{\subfigure[An ideal triangle-dense graph]{\epsfig{file=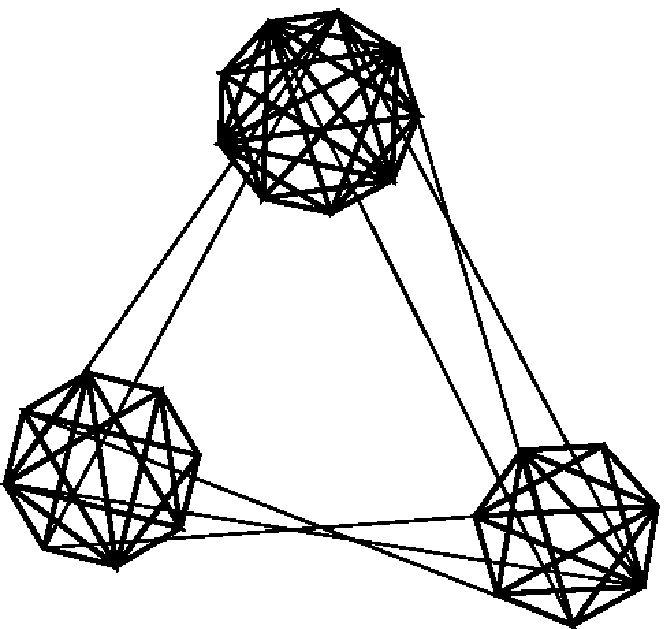,width=.325\textwidth}}\qquad\qquad
\subfigure[The lollipop graph]{\epsfig{file=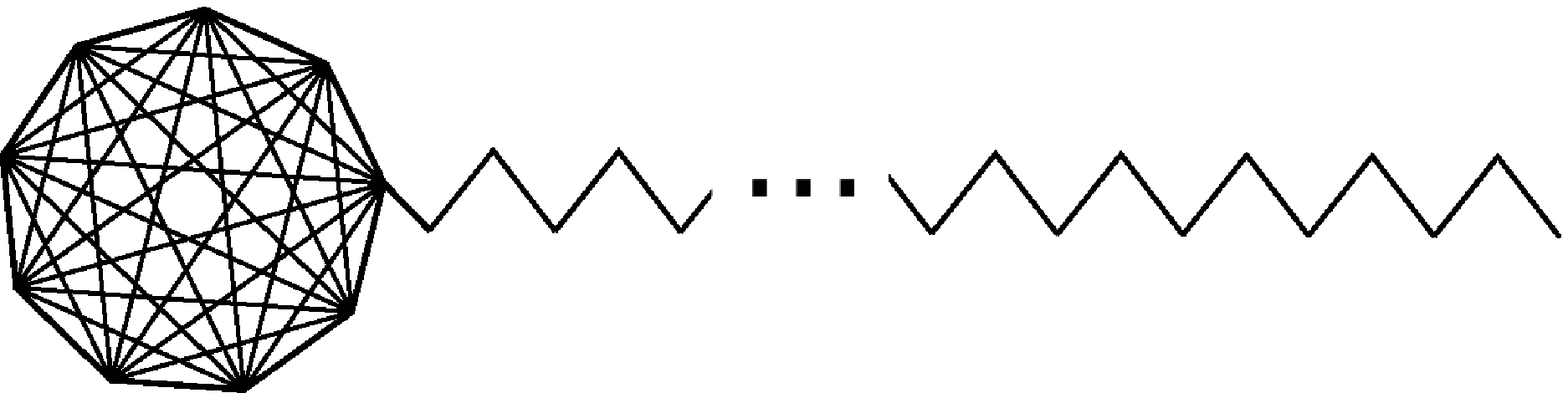,width = .425\textwidth}}}
\caption{Two examples of \td graphs with $\delta$ close to~1.}
\label{f:inverse}
\end{center}
\end{figure}

Nonetheless, an inverse theorem {\em does} hold if we redefine what it
means to approximate a graph by a collection of approximate cliques.
Instead of trying to capture most of the vertices or edges (which is
impossible, as the previous example shows), we consider the goal of
capturing a constant fraction of the {\em triangles} of a graph by a
collection of dense subgraphs.

\subsection{An Inverse Theorem}

To state an inverse theorem for triangle-dense graphs, we require a
preliminary definition.
\begin{definition}[Tightly Knit Family] \label{def:part}
Let $\rho > 0$.
A collection $V_1, V_2, \ldots, V_k$ of disjoint sets of vertices of a
graph $G=(V,E)$ forms a \emph{$\rho$-tightly-knit family} if:
\begin{enumerate}
	\item For each~$i=1,2,\ldots,k$, the subgraph induced by~$V_i$ has 
            at least $\rho \cdot \binom{|V_i|}{2}$ edges and $\rho \cdot \binom{|V_i|}{3}$ triangles.
            (That is, a $\rho$-fraction of the maximum possible edges and triangles.)
	\item For each~$i=1,2,\ldots,k$, the subgraph induced by~$V_i$
          has radius           at most $2$.
\end{enumerate}
\end{definition}

In Definition~\ref{def:part}, the vertex sets~$V_1,V_2,\ldots,V_k$ are
disjoint but need not cover all of~$V$; in particular, the empty
collection is technically a tightly knit family.

The following inverse theorem states that every
triangle-dense graph contains a tightly-knit family that captures most
of the ``meaningful social structure''---a constant fraction of the
graph's triangles.  
\begin{theorem}[\citet{GRSjournal}] \label{t:dense}
There is a function $f(\delta) = O(\delta^4)$ 
such that for every $\delta$-triangle dense graph~$G$, there exists an $f(\delta)$-tightly-knit family that contains an $f(\delta)$ fraction
of the triangles of $G$.
\end{theorem}

Graphs that are not triangle dense, such as
sparse Erd\H{o}s-R\'{e}nyi random graphs, do not generally admit 
$\rho$-tightly-knit families with constant $\rho$.
The complete tripartite graph shows that Theorem~\ref{t:dense} does
not hold if the ``radius-2'' condition in Definition~\ref{d:dense} is
strengthened to ``radius-1'' 
(Exercise~\ref{exer:tripartite}).

\subsection{Proof Sketch of Theorem~\ref{t:dense}}\label{ss:inversepf}

The proof of Theorem~\ref{t:dense} is constructive, and interleaves
two subroutines.  To state the first, define the {\em
  Jaccard similarity} of an edge~$(u,v)$ of a graph~$G$ as the
fraction of neighbors of $u$ and $v$ that are neighbors of both:
\[
\frac{|N(u) \cap N(v)|}{|N(u) \cup N(v)| - 2},
\]
where~$N(\cdot)$ denotes the neighbors of a vertex and the ``-2'' is
to avoid counting $u$ and~$v$ themselves.  The first subroutine, called the {\em
  cleaner}, is given a parameter $\epsilon$ as input and repeatedly
deletes edges with Jaccard similarity less than $\epsilon$ until none
remain.  Removing edges from the graph is worrisome because it removes
triangles, and Theorem~\ref{t:dense} promises that the final tightly
knit family captures a constant fraction of the original graph's
triangles.  But removing an edge with low Jaccard similarity destroys
many more wedges than triangles, and the number of triangles in the
graph is at least a constant fraction of the number of wedges (because
it is \td).  A charging argument along these lines shows that,
provided~$\epsilon$ is at most $\delta/4$, the cleaner cannot destroy
more than a constant fraction of the graph's triangles.

The second subroutine, called the {\em extractor}, is responsible for
extracting one of the clusters of the tightly-knit family from a graph in
which all edges have Jaccard similarity at least~$\epsilon$.  (Isolated
vertices can be discarded from further consideration.)  How is this
Jaccard similarity condition helpful?  One easy observation is that,
post-cleaning, the graph is ``approximately locally regular,'' meaning
that the endpoints of any edge have degrees within a
$\tfrac{1}{\epsilon}$ factor of each other.  Starting from this fact,
easy algebra shows that every one-hop neighborhood of the graph (i.e.,
the subgraph induced by a vertex and its neighbors) has
constant (depending on~$\epsilon$) density in both edges and triangles,
as required by Theorem~\ref{t:dense}.  The bad news is that extracting
a one-hop neighborhood can destroy almost all of a graph's triangles
(Exercise~\ref{exer:tripartite}). 
The good news is that supplementing
a one-hop neighborhood with a judiciously chosen subset of
the corresponding two-hop neighborhood (i.e., neighbors of neighbors)
fixes the problem.  Precisely, the extractor subroutine is
given a graph~$G$ in which every edge has Jaccard similarity at least
$\epsilon$ and proceeds as follows:
\begin{enumerate}

\item Let~$v$ be a vertex of~$G$ with the maximum degree.
  Let~$d_{max}$ denote $v$'s degree and~$N(v)$ its neighbors.
  
\item Calculate a score $\theta_w$ for every vertex $w$ outside $\{v\}
  \cup N(v)$ equal to the number of triangles that include~$w$ and
  two vertices of~$N(v)$.  In other words,
  $\theta_w$ is the number of triangles that would be saved by
  supplementing the
  one-hop neighborhood~$\{v\} \cup N(v)$ by~$w$.  (On the flip
  side, this  would also destroy the triangles that contain~$w$ and two
  vertices outside~$N(v)$.)

\item Return the union of $\{v\}$, $N(v)$, and up to~$d_{max}$ vertices
  outside $\{v\} \cup N(v)$ with the largest non-zero $\theta$-scores.

\end{enumerate}
It is clear that the extractor outputs a set~$S$ of vertices that
induces a subgraph with radius at most~2.  As with one-hop
neighborhoods, easy algebra shows that, because every edge has Jaccard
similarity at least $\epsilon$, this subgraph is dense in both edges
and triangles.  The important non-obvious fact, whose proof is omitted
here, is that the number of triangles saved by the extractor (i.e.,
triangles with all three vertices in its output) is at least a
constant fraction of the number of triangles it destroys (i.e.,
triangles with one or two vertices in its output).
It follows that alternating between cleaning and extracting (until no
edges remain) will produce a tightly-knit family meeting the promises
of Theorem~\ref{t:dense}.

\section{Power-Law Bounded Networks}\label{s:plb}

\newcommand{\Clm}[1]{{Claim~\ref{clm:#1}}} 
\newcommand{\Lem}[1]{{Lemma\,\ref{lem:#1}}} 
\newcommand{\Def}[1]{{Definition~\ref{def:#1}}} 

Arguably the most famous 
property of social and information networks, even more so than triadic
closure, is a
\emph{power-law degree distribution}, also referred to as a
heavy-tailed or scale-free degree distribution.

\subsection{Power-Law Degree Distributions and Their Properties}\label{ss:calcs}

Consider a simple graph $G = (V,E)$ with $n$
vertices.  For each positive integer $d$, let~$n(d)$ denote the number of
vertices of~$G$ with degree $d$. The sequence
$\{n(d)\}$ 
is called the \emph{degree distribution} of $G$. 
Informally, a degree distribution 
is said to be a \emph{power-law with exponent $\gamma > 0$} if 
$n(d)$ scales as $n/d^{\gamma}$.

There is some controversy about how to best fit power-law
distributions to data, and whether such distributions are the
``right" fit for the degree distributions in real-world social networks
(as opposed to, say, lognormal distributions).
Nevertheless, several of the consequences of a power-law degree
distribution assumption are uncontroversial for social networks,
and so a power-law distribution is a reasonable starting point for
mathematical analysis.

This section studies the algorithmic benefits of assuming that a graph
has an (approximately) power-law degree distribution, in the form of
fast algorithms for fundamental graph problems.  To
develop our intuition
about such graphs,
let's do some rough calculations under the assumption that $n(d) =
cn/d^\gamma$ (for some constant $c$) for every~$d$ up to the maximum
degree~$d_{max}$; think of $d_{max}$ as $n^{\beta}$ for some constant
$\beta \in (0,1)$.

First, we have the implication
\begin{equation}\label{eq:plb1}
\sum_{d \leq d_{max}} n(d) = n   \ \ \ \Longrightarrow \ \ \ cn \sum_{d \leq d_{max}} d^{-\gamma} = n.
\end{equation}

When $\gamma \leq 1$, $\sum_{d < \infty} d^{-\gamma}$ is a divergent
series. In this case, we cannot satisfy the right-hand side
of~\eqref{eq:plb1} with a constant $c$.
For this reason, results on power-law degree distributions typically
assume that $\gamma > 1$.  

Next, 
the  number of edges is exactly
\begin{equation}\label{eq:plb2}
\frac{1}{2} \sum_{d \leq d_{max}} d \cdot n(d) = \frac{cn}{2} \sum_{d \leq d_{max}}
d^{-\gamma+1}.
\end{equation}
Thus, up to constant factors, $\sum_{d \leq d_{max}} d^{-\gamma+1}$ is
the average degree.  For $\gamma > 2$,
$\sum_{d < \infty} d^{-\gamma+1}$ is a convergent series, and the
graph has constant average degree. For this reason, much of the early
literature on graphs with power-law degree distributions focused on
the regime where $\gamma > 2$. When $\gamma = 2$, the average degree
scales with $\log n$, and for $\gamma \in (1,2)$, 
it scales with $(d_{max})^{2-\gamma}$, which is polynomial in $n$.

One of the primary implications of a power-law degree distribution is
upper bounds on the number of high-degree vertices.
Specifically, under our assumption that $n(d) = cn/d^{\gamma}$,
the number of vertices of degree \emph{at least} $k$ can be bounded by
\begin{equation}
\sum_{d = k}^{d_{max}} n(d) \leq cn \sum_{d = k}^\infty d^{-\gamma}
\leq cn \int^\infty_{k} x^{-\gamma} \, dx = cnk^{-\gamma+1}/(\gamma-1) = \Theta(n k^{-\gamma+1}). \label{eq:tail}
\end{equation}

\subsection{PLB Graphs}

The key definition in this section is a more plausible and robust
version of the assumption that $n(d) = cn/d^{\gamma}$, for which
the conclusions of calculations like those in Section~\ref{ss:calcs}
remain valid.  The definition allows individual values of $n(d)$ to
deviate from a true power law, while requiring (essentially) that the
average value of $n(d)$ in sufficiently large intervals of $d$ does
follow a power law.

\begin{definition}[\citet{B+15,plb}] \label{def:plb} A graph $G$  with degree
  distribution $\{ n(d) \}$ is a \emph{power-law bounded (PLB) graph
  with exponent $\gamma > 1$} if there is a constant $c > 0$ such
that
\[
\sum_{d = 2^r}^{2^{r+1}} n(d) \leq cn \sum_{d = 2^r}^{2^{r+1}}
  d^{-\gamma}
\]
for all $r \ge 0$.
\end{definition}
Many real-world social networks satisfy a mild generalization of this
definition, in which $n(d)$ is allowed to scale with
$n/(d+t)^{\gamma}$ for a ``shift''~$t \ge 0$; see the Notes for details.
For simplicity, we continue to assume in this section that $t=0$.

Definition~\ref{def:plb} has several of the same implications as a
pure power law assumption, including the following lemma
(cf.~\eqref{eq:plb2}).

\begin{lemma} \label{lem:plb} Suppose $G$ is a PLB graph with exponent
  $\gamma > 1$. For every $c > 0$ and natural  number $k$,
\[
\sum_{d \leq k} d^c \cdot n(d) = O\left(n \sum_{d \leq k} d^{c-\gamma}\right).
\]
\end{lemma}

The proof of Lemma~\ref{lem:plb} is technical but not overly
difficult; we do not discuss the details here.

The first part of the next lemma provides control over the number of  high-degree vertices
and is the primary reason why many graph problems are more easily
solved on PLB graphs than on general graphs.  The second part of the
lemma bounds the number of wedges of the graph when $\gamma \ge 3$.

\begin{lemma} \label{clm:plb_bound} Suppose $G$ is a PLB graph with exponent $\gamma > 1$. Then:
\begin{itemize}
    \item [(a)] $\sum_{d \geq k} n(d) = O(n k^{-\gamma+1})$.
    \item [(b)] Let $W$ denote the number of wedges (i.e., two-hop paths).
If $\gamma = 3$, $W = O(n\log n)$. If $\gamma > 3$, $W = O(n)$.
\end{itemize}
\end{lemma}

Part~(a) extends the computation in~\eqref{eq:tail} to PLB graphs, while
part~(b) follows from Lemma~\ref{lem:plb} (see Exercise~\ref{exer:plb}).

\subsection{Counting Triangles} \label{sec:tri_count} 

Many graph problems appear to be easier in PLB graphs than in general
graphs.  To illustrate this point, we single out the problem of {\em
  triangle counting}, which is one of the most canonical problems in
social network analysis.
For this section, we assume that our algorithms can determine in constant
time if there is an edge between a given pair of vertices;
these lookups can be avoided with a careful 
implementation (Exercise~\ref{exer:lookup}), but such details distract
from the main analysis.

As a warm up, consider the following trivial algorithm to count
(three times) the number of triangles of a given graph~$G$
(``Algorithm~1''):
\begin{itemize}

\item For every vertex $u$ of~$G$:

\begin{itemize} 

\item For every pair $v,w$ of $u$'s neighbors, check if
  $u$, $v$, and $w$ form a triangle.

\end{itemize}

\end{itemize}
Note that the running time of Algorithm~1 is proportional to the
number of wedges in the graph~$G$.
The following running time bound for triangle counting in PLB graphs
is an immediate corollary of Lemma~\ref{clm:plb_bound}(b), applied to
Algorithm~1.

\begin{corollary}
  Triangle counting in $n$-vertex PLB graphs with exponent $3$ can be
  carried out in $O(n\log n)$ time.  If the exponent is strictly
  greater than $3$, it can be done in $O(n)$ time.
\end{corollary}

Now consider an optimization of Algorithm~1 (``Algorithm~2''):
\begin{itemize}

\item Direct each edge of $G$ from the lower-degree endpoint to the
  higher-degree endpoint (breaking ties lexicographically) to obtain a
  directed graph~$D$.

\item For every vertex $u$ of~$D$:

\begin{itemize} 

\item For every pair $v,w$ of $u$'s {\em out-neighbors}, check if
  $u$, $v$, and $w$ form a triangle in~$G$.

\end{itemize}

\end{itemize}
Each triangle is counted exactly once by Algorithm~2, in the
iteration where the lowest-degree of its three vertices plays the role
of~$u$.
Remarkably, this simple idea leads to massive time savings in
practice.

A classical way to capture this running time improvement
mathematically is to parameterize the input graph~$G$ by its {\em
  degeneracy}, which can be thought of as a refinement of
the maximum degree.  The degeneracy $\alpha(G)$ of a graph~$G$ can be
computed by iteratively removing a minimum-degree vertex (updating
the vertex degrees after each iteration) until no
vertices remain; $\alpha(G)$ is then the largest degree of a 
vertex at the time of its removal.  (For example, every tree has
degeneracy equal to~1.)  We have the following
guarantee for Algorithm~2, parameterized by a graph's degeneracy:
\begin{theorem}[\citet{ChNi85}]
For every graph with~$m$ edges and degeneracy~$\alpha$, the running
time of Algorithm~2 is~$O(m \alpha)$.
\end{theorem}

Every PLB graph with exponent $\gamma > 1$ has degeneracy
$\alpha = O(n^{1/\gamma})$; see Exercise~\ref{exer:plbdegen}.  For
PLB graphs with~$\gamma > 2$, we can apply Lemma~\ref{lem:plb} with $c=1$
to obtain $m=O(n)$ and hence the running time of Algorithm~2 is
$O(m\alpha) = O(n^{(\gamma+1)/\gamma})$.

Our final result for PLB graphs improves this running time bound, for
all $\gamma \in (2,3)$, through a more refined analysis.\footnote{The
  running time bound
actually holds for all $\gamma \in (1,3)$, but is an improvement only
for $\gamma > 2$.}

\begin{theorem}[\citet{plb}] \label{thm:tricount_cn} 
In PLB graphs with exponent $\gamma \in (2,3)$, 
Algorithm~2 runs in $O(n^{3/\gamma})$ time.
\end{theorem}

\begin{proof} 
Let~$G=(V,E)$ denote an $n$-vertex PLB graph with exponent $\gamma \in
(2,3)$.  Denote the degree of vertex $v$ in $G$ by $d_v$ and its
out-degree in the directed graph~$D$ by~$d^+_v$.
The running time of Algorithm~2 is $O(n+ \sum_v {d^+_v \choose 2}) =
O(n + \sum_v (d^+_v)^2)$, so the analysis boils down to
bounding the out-degrees in~$D$. 
One trivial upper bound is $d^+_v \leq d_v$ for every~$v \in V$. 
Because every edge is directed from its lower-degree endpoint to its
higher-degree endpoint, we also have
$d^+_v \leq \sum_{d \geq d_v} n(d)$.
By \Clm{plb_bound}(a), the second bound is $O(n d^{-\gamma+1}_v)$. 
The second bound is better than the first roughly when $d_v \geq n
d^{-\gamma+1}_v$, or equivalently when $d_v \geq n^{1/\gamma}$.

Let $V(d)$ denote the set of degree-$d$ vertices of~$G$. We split the
sum over vertices according to how their degrees compare to $n^{1/\gamma}$,
using the first bound for low-degree vertices and the second bound for
high-degree vertices:
\begin{eqnarray*}
\sum_{v \in V} (d^+_v)^2 & = & \sum_d \sum_{v \in V(d)} (d^+_v)^2\\
 & \leq & \sum_{d \leq n^{1/\gamma}} \sum_{v \in V(d)} d^2 + \sum_{d > n^{1/\gamma}} \sum_{v \in V(d)} O(n^2 d^{-2\gamma+2}) \nonumber \\
& = & \sum_{d \leq n^{1/\gamma}} d^2 \cdot n(d) + O\left(n^2 \cdot \sum_{d
      > n^{1/\gamma}} d^{-2\gamma+2} \cdot n(d)\right).
\end{eqnarray*}

Applying \Lem{plb} (with $c=2$) to the sum over low-degree vertices,
and using the fact that with $\gamma < 3$ the sum $\sum_d d^{2-\gamma}$ is
divergent, we derive
\[
\sum_{d \leq n^{1/\gamma}} d^2 \cdot n(d) = O\left(n \sum_{d \leq n^{1/\gamma}} d^{2-\gamma}\right) = O(n (n^{1/\gamma})^{3-\gamma}) = O(n^{3/\gamma}).
\]

The second sum is over the highest-degree vertices, and \Lem{plb} does
not apply. On the other hand, we can invoke \Clm{plb_bound}(a) to obtain
the desired bound:
\begin{eqnarray*}
n^2 \sum_{d > n^{1/\gamma}} d^{-2\gamma+2} \cdot n(d) & \leq & n^2 (n^{1/\gamma})^{-2\gamma+2} \sum_{d > n^{1/\gamma}} n(d)\\
& = & O(n^{2/\gamma} \cdot n(n^{1/\gamma})^{-\gamma+1})\\
& = & O(n^{3/\gamma}).
\end{eqnarray*}
\end{proof}

The same reasoning shows that Algorithm~2 runs in~$O(n \log n)$ time
in $n$-vertex PLB graphs with exponent $\gamma=3$, and in~$O(n)$ time
in PLB graphs with $\gamma > 3$ (Exercise~\ref{exer:gamma3}).

\subsection{Discussion}\label{ss:disc}

Beyond triangle counting, which computational problems should we
expect to be easier on
PLB graphs than on general graphs?  A good starting point is problems that
are relatively easy on bounded-degree graphs.  In many cases,
fast algorithms for bounded-degree graphs remain fast for graphs with
bounded degeneracy.  In these cases, the degeneracy bound for PLB
graphs (Exercise~\ref{exer:plbdegen}) can already lead to fast
algorithms for such graphs.  For example, this approach can be used to
show that all of the cliques of a PLB graph with exponent $\gamma > 1$
can be enumerated in subexponential time (see Exercise~\ref{exer:clique}).
In some cases, like in Theorem~\ref{thm:tricount_cn}, one can beat the
bound from the degeneracy-based analysis through more refined
arguments.

\section{The BCT Model}\label{s:bct}

\newcommand{\ecc}{\mathrm{ecc}}
\newcommand{\twosweep}{{\tt TwoSweep}}
\newcommand{\dist}{\mathrm{dist}}
\newcommand{\eqdef}{:=}

This section gives an impressionistic overview of another set of
deterministic conditions meant to capture properties of ``typical
networks,'' proposed by \citet{BCT17} and hereafter called the {\em
  BCT model}.  The precise model is 
technical with a number of parameters; we give only a high-level
description that ignores several complications. 

To illustrate the main ideas, consider the problem of computing the
{\em diameter} $\max_{u,v \in V} \dist(u,v)$ of an undirected and
unweighted $n$-vertex graph~$G=(V,E)$, where $\dist(u,v)$ denotes the
shortest-path distance between $u$ and $v$ in $G$.  Define the
\emph{eccentricity} of a vertex~$u$ by
$\ecc(u) \eqdef \max_{v \in V} \dist(u,v)$, so that the diameter is
the maximum eccentricity.  The eccentricity of a single vertex can be
computed in linear time using breadth-first search, which gives a
quadratic-time algorithm for computing the diameter.
Despite much 
effort, no subquadratic $(1+\eps)$-approximation algorithm
for computing the graph diameter is known for general graphs. 
Yet there are many heuristics that
perform well in real-world networks. 
Most of these heuristics compute the
eccentricities of a carefully chosen subset of vertices. An extreme
example is the \twosweep{} algorithm:
\begin{enumerate}

\item  Pick an arbitrary vertex $s$,
and 
perform breadth-first search from~$s$ to compute a vertex
$t \in \argmax_{v \in V} \dist(s,v)$. 

\item Use breadth-first search again to compute $\ecc(t)$ and return
  the result.

\end{enumerate}
This heuristic always produces a lower bound on a graph's diameter,
and in practice usually achieves a close approximation.
What properties of ``real-world'' graphs might explain this empirical performance?

The BCT model is largely inspired by the metric properties of random
graphs.
To explain, for a vertex $s$ and natural number $k$, let $\tau_s(k)$
denote the smallest length~$\ell$ so that there are at least $k$
vertices at distance (exactly) $\ell$ from $s$.  Ignoring the
specifics of the random graph model, the $\ell$-step neighborhoods (i.e.,
vertices at distance exactly~$\ell$) of a vertex in a random graph
resemble uniform random sets of size increasing with $\ell$.  We next use
this property to derive a heuristic
upper bound on $\dist(s,t)$.  Define $\ell_s \eqdef \tau_s(\sqrt{n})$ and
$\ell_t \eqdef \tau_t(\sqrt{n})$. Since the $\ell_s$-step neighborhood
of $s$ and the $\ell_t$-step neighborhood of $t$ act like random sets of size
$\sqrt{n}$, a birthday paradox argument implies that they intersect
with non-trivial probability.  If they do intersect, then $\ell_s +
\ell_t$ is an upper bound on $\dist(s,t)$.
In any event, we can adopt this inequality
as a deterministic graph property, which
can be tested against real network data.\footnote{The actual BCT model
  uses the upper bound $\tau_s(n^{x}) +
  \tau_t(n^{y})$ for $x + y > 1+\delta$, to ensure
  intersection with high enough
  probability.}

\begin{property}\label{prop:1}
For all $s,t \in V$, $\dist(s,t) \leq
  \tau_s(\sqrt{n}) + \tau_t(\sqrt{n})$.
\end{property}

One would expect this distance upper bound to be tight for pairs of
vertices that are far away from each other, and in a reasonably random
graph, this will be true for most of the vertex pairs.
This leads us to the next property.\footnote{We omit the exact definition
of this property in the BCT model, which is quite involved.}

\begin{property}\label{prop:2}
For all $s \in V$: for ``most" $t \in V$, $\dist(s,t) > \tau_s(\sqrt{n}) + \tau_t(\sqrt{n}) - 1$.
\end{property}

The third property posits a distribution on the $\tau_s(\sqrt{n})$
values. Let $T(k)$ denote the average $n^{-1} \sum_{s \in V} \tau_s(k)$.

\begin{property}\label{prop:3}
There are constants $c,\gamma > 0$ such that
the fraction of vertices~$s$ satisfying $\tau_s(\sqrt{n}) \geq
T(\sqrt{n}) + \gamma$ is roughly $c^{-\gamma}$.
\end{property}
A consequence of this property is that the largest value of
$\tau_s(\sqrt{n})$ is $T(\sqrt{n}) + \log_c n + \Theta(1)$.

As we discuss below, these properties will imply that simple
heuristics work well for computing the diameter of a graph. On the
other hand, these properties do not generally hold in real-world
graphs. The actual BCT model has a nuanced version of these
properties, parameterized by vertex degrees. In addition, the BCT
model imposes an approximate power-law degree distribution, in the
spirit of power-law bounded graphs (Definition~\ref{def:plb} in
Section~\ref{s:plb}). This nuanced list of properties can be
empirically verified on a large set of real-world graphs.

Nonetheless, for understanding the connection
of metric properties to diameter computation,
it suffices to look at Properties~\ref{prop:1}--\ref{prop:3}. We can now bound the
eccentricities of vertices. The properties imply that
\[
\dist(u,v) \leq \tau_u(\sqrt{n}) + \tau_v(\sqrt{n}) \leq
\tau_u(\sqrt{n}) + T(\sqrt{n}) + \log_c n + O(1).
\]
Fix $u$ and imagine varying $v$ to estimate $\ecc(u)$.  For ``most"
vertices $v$,
$\dist(u,v) \geq \tau_u(\sqrt{n}) + \tau_v(\sqrt{n}) - 1$.  By Property~\ref{prop:3},
one of the
vertices~$v$ satisfying this lower bound will also satisfy
$\tau_v(\sqrt{n}) \geq T(\sqrt{n}) + \log_c n - \Theta(1)$. Combining,
we can bound the eccentricity by
\begin{equation}\label{eq:ecc}
\ecc(u) = \max_v \dist(u,v) = \tau_u(\sqrt{n}) + T(\sqrt{n}) + \log_c
n \pm \Theta(1).
\end{equation}
The bound~\eqref{eq:ecc} is significant because it reduces maximizing
$\ecc(u)$ over $u \in V$ to maximizing $\tau_u(\sqrt{n})$.

Pick an arbitrary vertex $s$ and consider a vertex $u$ that maximizes
$\dist(s,u)$. By an argument similar to the one above (and because
most vertices are far away from~$s$), we expect that
$\dist(s,u) \approx \tau_s(\sqrt{n}) + \tau_u(\sqrt{n})$. Thus, a
vertex $u$ maximizing $\dist(s,u)$ is almost the same as a vertex
maximizing $\tau_u(\sqrt{n})$, which by~\eqref{eq:ecc} is almost the
same as a vertex maximizing $\ecc(u)$. This gives an explanation of why
the \twosweep{} algorithm performs so well.  Its first use of
breadth-first search identifies a vertex~$u$ that (almost) maximizes
$\ecc(u)$. The second pass of breadth-first search (from~$u$) then
computes a close approximation of the diameter.

The analysis in this section is heuristic, but it captures much of the
spirit of algorithm analysis in the BCT model.  These results for
\twosweep{} can be extended to other heuristics that choose a set of
vertices through a random process to lower bound the diameter.  In
general, the key insight is that most distances $\dist(u,v)$ in the
BCT model can be closely approximated as a sum of quantities that
depend only on either $u$ or $v$.

\section{Discussion}

Let's take a bird's-eye view of this chapter.  The big challenge in
the line of research described in this chapter is the formulation of
graph classes and properties
that both reflect real-world graphs and lead
to a satisfying theory. 
It seems unlikely that any one class of graphs will simultaneously
capture all the relevant properties of (say) social networks. 
Accordingly, this chapter described several graph classes that 
target specific empirically observed graph properties, each with
its own algorithmic lessons:
\begin{itemize}
    \item Triadic closure aids the computation of dense
      subgraphs.
    \item Power-law degree distributions aid subgraph counting.
    \item $\ell$-hop neighborhood structure influences the structure
      of shortest paths. 
\end{itemize}
These lessons suggest that, when defining a graph class to capture 
``real-world'' graphs, it may be important to keep a target
algorithmic application in mind.

Different graph classes differ in how closely the
definitions are tied to domain knowledge and empirically observed
statistics.
The $c$-closed and triangle-dense graph classes are in the spirit of
classical families of graphs (e.g., planar or bounded-treewidth
graphs), and they sacrifice precision in the service of generality,
cleaner definitions, and arguably more elegant theory.
The PLB and BCT frameworks take the opposite view: the graph
properties are quite technical and involve many parameters, and in
exchange tightly capture the properties of ``real-world''
graphs. These additional details can add fidelity to theoretical
explanations for the surprising effectiveness of simple heuristics.

A big advantage of combinatorially defined graph classes---a hallmark
of graph-theoretic work in theoretical computer science---is the ability
to empirically validate them on real data.  The standard statistical
viewpoint taken in network science has led to dozens of competing
generative models,
and it is nearly impossible to validate the details of 
such a model from network data.
The deterministic graph classes defined in this chapter give a much
more satisfying foundation for algorithmics on real-world graphs.

Complex algorithms for real-world problems can be useful, but
practical algorithms for graph analysis are typically based on simple
ideas like backtracking or greedy algorithms.  An ideal theory would
reflect this reality, offering compelling explanations 
for why relatively simple algorithms have such
surprising efficacy in practice.

We conclude this section with some open problems.

\begin{enumerate}

\item Theorem~\ref{t:c-closed} gives, for constant~$c$, a bound of
  $O(n^2)$ on the number of maximal cliques in a $c$-closed
  graph. \citet{cclosed} also prove a sharper bound of
  $O(n^{2(1-2^{-c})})$, which is asymptotically tight when
  $c=2$. Is it tight for all values of~$c$?  Additionally,
  parameterizing by the number of edges ($m$) rather than vertices~($n$), is the number of maximal cliques in a $c$-closed graph with
  $c=O(1)$ bounded by $O(m)$?  Could there be a linear-time algorithm
  for maximal clique enumeration for $c$-closed graphs with constant
  $c$?

\item Theorem~\ref{t:dense} guarantees the capture by a tightly-knit
  family of an~$O(\delta^4)$
  fraction of the triangles of a \td graph.  What is the best-possible
  constant in the exponent?  Can the upper bound be improved, perhaps
  under additional assumptions (e.g., about the distribution of the
  clustering coefficients of the graph, rather than merely about their
  average)?

\item \citet{UgBaKl13} observe that $4$-vertex subgraph counts in
  real-world graphs exhibit predictable and peculiar behavior. By
imposing conditions on 4-vertex subgraph counts
(in addition to triangle density), can one prove
decomposition theorems better than Theorem~\ref{t:dense}?

\item Is there a compelling algorithmic application for graphs that
can be approximated by tightly-knit families?

\item \citet{BeGlLe16} and \citet{TsPaMi17} defined the \emph{triangle
    conductance} of a graph, where cuts are measured in terms of the
  number of triangles cut (rather than the number of edges). Empirical
  evidence suggests that cuts with low triangle conductance give more
  meaningful communities (i.e., denser subgraphs) than cuts with low
  (edge) conductance.  Is there a plausible theoretical
  explanation for this observation?

\item A more open-ended goal is to use the theoretical insights
  described in this chapter to develop new and practical algorithms for
fundamental graph problems.

\end{enumerate}

\section{Notes}

The book by \citet{EK} is a good introduction to social networks
analysis, including discussions of heavy-tailed degree
distributions and triadic closure.
A good if somewhat outdated review of generative models for social and
information networks is \citet{ChFa06}.
The Enron email network was first studied by \citet{KY04}.

The definitions of $c$-closed and weakly $c$-closed graphs
(Definitions~\ref{d:closed}--\ref{d:weak}) are from
\cite{cclosed}, as is the fixed-parameter tractability result for the 
maximum clique problem (Theorem~\ref{t:c-closed}).  
\citet{ELS10} proved an analogous result
with respect to a different parameter, 
the degeneracy of the input graph.
The reduction from efficiently enumerating maximal cliques to bounding the
number of maximal cliques (Theorem~\ref{t:enum}) is from
\citet{T+77}.  Moon-Moser graphs and the Moon-Moser bound on the
maximum number of maximal cliques of a graph are from \citet{MM65}.

The definition of triangle-dense graphs (Definition~\ref{d:dense}) and
the inverse theorem for them (Theorem~\ref{t:dense}) are from
\citet{GRSjournal}.  The computation of the triangle density of the
Facebook graph is detailed by~\citet{UgKa+11}.

The definition of power law bounded graphs (Definition~\ref{def:plb})
first appeared in \citet{B+15} in the context of triangle counting,
but it was formalized and applied to many different problems by
\citet{plb}, including triangle counting
(Theorem~\ref{thm:tricount_cn}), clique enumeration
(Exercise~\ref{exer:clique}), and linear algebraic problems for matrices
with a pattern of non-zeroes that induces a PLB graph.
\citet{plb} also performed a detailed empirical analysis, validating
\Def{plb} (with small shifts~$t$) on real data. 
The degeneracy-parameterized bound for counting triangles is
essentially due to \cite{ChNi85}.

The BCT model (Section~\ref{s:bct}) and the fast algorithm for
computing the diameter of a graph are due to \citet{BCT17}.

\section*{Acknowledgments}

The authors thank Michele Borassi, Shweta Jain, Piotr Sankowski, and
Inbal Talgam-Cohen for their comments on earlier drafts of this
chapter.

\section*{Exercises}

\begin{enumerate}

\item \label{exer:elim}
Prove that a graph is weakly $c$-closed in the sense of
Definition~\ref{d:weak} if and only if its vertices can be ordered
$v_1,v_2,\ldots,v_n$ such that, for every $i=1,2,\ldots,n$, the
vertex~$v_i$ is $c$-good in the subgraph induced by
$v_i,v_{i+1},\ldots,v_n$.

\item \label{exer:maximal}
Prove that the backtracking algorithm in Section~\ref{ss:backtracking}
enumerates all of the maximal cliques of a graph.

\item \label{exer:density} 
Prove that a graph has triangle density $1$ if and only if it is a disjoint
union of cliques.

\item \label{exer:tripartite} Let~$G$ be the complete regular
  tripartite graph with~$n$ vertices---three vertex sets of size
  $\tfrac{n}{3}$ each, with each vertex connected to every vertex of
  the other two groups and none of the vertices within the same group.
\begin{exlist}
\item What is the triangle density of the graph?

\item What is the output of the cleaner (Section~\ref{ss:inversepf})
 when applied to this graph?  What is then the output of the extractor?

\item Prove that $G$ admits no tightly-knit family that contains a
  constant fraction (as $n \rightarrow \infty$) of the graph's
  triangles and uses only radius-1 clusters.

\end{exlist}

\item \label{exer:plb}
Prove \Clm{plb_bound}. 

\vspace{.25\baselineskip}
\noindent
[Hint: To prove~(a), break up the sum over degrees
into sub-sums between powers of $2$. Apply \Def{plb} to each sub-sum.]

\item \label{exer:lookup}
Implement Algorithm~2 from
Section~\ref{sec:tri_count}
in $O(\sum_v (d^+_v)^2 + n)$ time, where~$d^+_v$ is the
number of out-neighbors of~$v$ in the directed version~$D$ of~$G$,
assuming that the input $G$ is represented using only adjacency lists. 

\vspace{.25\baselineskip}
\noindent
[Hint: you may need to store the in- and out-neighbor lists of~$D$.]

\item \label{ex:degen} Prove that every graph with $m$ edges has
  degeneracy at most $\sqrt{2m}$.  Exhibit a family of graphs showing
  that this bound is tight (up to lower order terms).

\item \label{exer:plbdegen}
Suppose $G$ is a PLB graph with exponent $\gamma > 1$.
\begin{exlist}
    \item Prove that the maximum degree of $G$ is $O(n^{1/(\gamma-1)})$.
    \item Prove that the degeneracy is $O(n^{1/\gamma})$. 
\end{exlist}

\noindent
\vspace{.25\baselineskip}
[Hint: For~(b), use the main idea in the proof of
Exercise~\ref{ex:degen} and \Clm{plb_bound}.]

\item \label{exer:gamma3}
Prove that Algorithm~2 in Section~\ref{sec:tri_count} runs
  in~$O(n \log n)$ time and~$O(n)$ time in $n$-vertex PLB graphs with
  exponents $\gamma=3$ and $\gamma > 3$, respectively.

\item \label{exer:clique}
Prove that all of the cliques of a graph with degeneracy~$\alpha$ can
be enumerated in $O(n2^{\alpha})$ time. 
(By Exercise~\ref{exer:plbdegen}(b), this immediately gives a
subexponential-time algorithm for enumerating the cliques of a PLB graph.)

\end{enumerate}

\end{document}